\documentclass[runningheads]{llncs}

\usepackage{iftex}
\usepackage[utf8]{inputenc}
\usepackage{amsmath}
\usepackage{amssymb}
\usepackage{amsfonts}
\usepackage{stmaryrd}
\usepackage{xspace}
\usepackage{proof}
\usepackage{tikz}
\usepackage{tikz-cd}
\usepackage{proof}
\usepackage{url}
\include{quiver}

\ifpdf
  \usepackage{underscore}         
  \usepackage[T1]{fontenc}        
\else
  \usepackage{breakurl}           
\fi

\newcommand{\Int}[2]{\mathsf{Int}(#1,#2)}
\newcommand{\ra}{\rightarrow}
\newcommand{\impl}[1][]{\overset{#1}{\Longrightarrow}}
\newcommand{\coimpl}[1][]{\overset{#1}{\Longleftarrow}}
\newcommand{\wand}{-\!\!\!\ast}

\newcommand{\dast}{\overset{\rightarrow}{\ast}}
\newcommand{\disast}{\ast_d}
\newcommand{\dwand}{\overset{\rightarrow}{\wand}}

\newcommand{\comp}{\otimes}
\newcommand{\dcomp}{\overset{\rightarrow}{\comp}}
\newcommand{\discomp}{\comp_d}

\newcommand{\dom}{\mathrm{dom}}

\newcommand{\Comp}[1]{\mathrm{Comp}(#1)}
\newcommand{\Beh}[1]{\mathrm{Beh}(#1)}

\newcommand{\C}{\mathbb{C}}

\newcommand{\X}{\mathcal{X}}

\newcommand{\comment}[1]{}

\newcommand{\val}{\mathcal{V}}

\newcommand{\goes}[1]{\stackrel{#1}{\rightarrow}}

\begin{document}

\title{Minimalistic System Modelling: 
    Behaviours, \\ Interfaces, and Local Reasoning }

\author{Didier Galmiche \inst{1}\orcidID{0000-0002-9968-5323} \and
Timo Lang \inst{2}\orcidID{0000-0002-8257-968X} \and 
David Pym \inst{2,3}\orcidID{0000-0002-6504-5838} 
}

\authorrunning{Galmiche, Lang, and Pym}

\institute{Universit\'e de Lorraine, CNRS, LORIA, Nancy, France \email{didier.galmiche@loria.fr} \and 
University College London, UK \email{\{timo.lang,d.pym\}@ucl.ac.uk} \and 
Institute of Philosophy, University of London, UK \email{david.pym@sas.ac.uk}}

\maketitle

\begin{abstract}
The infrastructure upon which the functioning of society depends is composed of complex ecosystems of systems. Consequently, we must reason about the properties of such ecosystems, which requires that we construct models of them. There are very many approaches to systems modelling, typically building on complex structural and dynamic frameworks. For example, there are simulation modelling tools based on a theoretical treatment of a 
`distributed systems metaphor' for system modelling. These tools are based on the concepts of location, resource, and process. They are well developed
and have proved valuable in a range of settings. Arguably, however, the foundations of modelling technologies remain underdeveloped. Our purpose here is to explore a foundational modelling framework based on minimal assumptions, starting from a primitive notion of behaviour, and to show that such an approach allows the recovery of the key ideas, including a generalized CAP theorem, required for effective modelling of and reasoning about ecosystems of systems. We establish a logic of behaviours and use it to express local reasoning principles for the compositional structure of systems. 
\end{abstract}

\section{Introduction} \label{sec:introduction}

The aim of this paper is to propose and study an extremely simple but general framework for 
modelling \emph{distributed systems}, seen as a metaphor for ecosystems 
of systems. We stress that what we describe is a foundational conceptual framework for constructing models, not a specific model. 

There is a good deal of work on systems modelling in the literature, 
much of it --- such as \cite{Hillston96} and \cite{CP09,CMP10,CMP2012,CP15,AP16}, for example 
--- deriving from process algebra. Some of this work, such as \cite{CP15,CIP2022}, 
considers composition of models via notions of interface. Other work, such as 
\cite{SystemDynamics}, derives from the world of management science and 
organizational theory. Other approaches are based more immediately 
on notions of composition and interface and their dynamics. For example, 
interface automata \cite{AH2001b,AH2001a}, relational interfaces \cite{Tripakis2011}, 
and bigraphs, \cite{Milner2002,Milner2008}. Our approach differs from all of these 
(and others) in its desire to employ minimal assumptions about structure and dynamics, 
beginning with a very primitive notion of behaviour. Our aim is to explore the (arguably underdeveloped) \emph{foundations} of system modelling, not to replace existing frameworks and tools.

For the practical purposes of simulation modelling --- see, for example, \cite{BoK-MS} for a survey of the overall ideas --- behaviour can be captured, for example, by the dynamics of the many forms of process algebra. Examples include 
\cite{Milner1989,Hoare85,CP09,CMP10,CMP2012,AP16,caulfield2016mod,CIP2022,Hillston96} and many more, including those, described in various previous Simutools papers \cite{CMP10,caulfield2016mod,CIP2022} that are intended to support implementations of the `distributed systems metaphor' --- in which a system is modelled in terms of its \emph{locations}, \emph{resources}, and \emph{processes}, relative to its \emph{environment} --- as described in \cite{CP09,CMP10,CMP2012,AP16,caulfield2016mod,CIP2022}. The history of this work can be traced to Birtwistle's `Demos: Discrete Event Modelling on Simula' \cite{Bir79}. 

These approaches can be implemented by tools such as those described in languages such as Gnosis \cite{CMP10,CMP2012} and Julia \cite{BKSE2012,julialang}. In particular, the modelling approach based on the `distributed systems metaphor', is conveniently captured in Julia using the `SysModels' package that is  available from GitHub \cite{juliacode}. 

Indeed, the distributed systems metaphor notwithstanding, there are very many --- far too numerous to discuss here --- informal definitions of systems discussed across many disciplines. However, we can identify a common core of these definitions,  all supported by the distributed systems metaphor and its implementations, that will serve as a starting point for our work:  
\begin{quote}
    A system is something whose behaviour is reflected in the behaviour of its components.
\end{quote}
This definition is deceptively simple: it identifies, at least implicitly, not only \emph{behaviour} and \emph{components}, but also 
\emph{composition} (necessarily, of components) and so (again, necessarily) \emph{interfaces} between components. This definition amounts to an abstraction and generalization of the characterization of systems that is provided by the distributed systems metaphor (e.g., \cite{CP15,CIP2022}). 

In our framework, then, the notion of a `behaviour' will be a primitive. 
The standard model of a behaviour is a trace (finite or infinite), 
but we do not have to subscribe to this interpretation. Indeed, our 
primitive notion of behaviour requires not even a commitment to that 
degree of structure, although traces will be an important example. 
The architectural aspect of a system, namely its component  subsystems, will be considered a priori knowledge.

 
We will see that from this simple starting point --- that is, a primitive notion of behaviour --- a rather expressive theory 
emerges. In Section~\ref{sec:basic}, we introduce the main notions of our modelling approach. Essentially, we define what we consider to be a `system', and this will built on a primitive notion of \emph{behaviour}. After defining the crucial notion of an \emph{implementation} between systems, we then derive further well-known concepts from system modelling; in particular, \emph{interfaces} and \emph{composition}.

A thorough epistemological investigation of our choices of primitives is beyond the scope of this paper. However, we believe it is evident that we have identified a credible minimalistic position. 

In Section~\ref{sec:CAP}, by way of demonstrating the relevance of our approach, we use our framework to establish an abstract statement 
of the famous CAP theorem \cite{Brewer2012,Lynch2012}. Brewer's original conjecture \cite{brewer2000} is 
expressed rather generally. It concerns three key properties of distributed systems. We suppose 
a system with multiple, connected components that is intended to deliver a specified 
service. Informally:
\begin{itemize}
    \item[-] \emph{(C)onsistency}: all the system's clients see the same data at the 
same time, irrespective of the component of the system to which they connect; 
\item[-] \emph{(A)vailability}: every request for the service that is sent to the system receives 
a response that is not the `error' response;
\item[-] \emph{(P)artition-tolerance}: the system 
continues to deliver the service even if there are failures of the connections between 
its components.
\end{itemize}
Brewer stated that these three guarantees cannot be satisfied simultaneously in a distributed system. Gilbert and Lynch gave a formal formulation and proof \cite{GL2002} 
in a set-up that is quite concrete. Specifically, they give a formal statement of Brewer's conjecture 
for an account of the CAP properties in `web services'. In our framework, we are able 
to establish the CAP theorem in a conceptually quite general way. However, we show, in Appendix~\ref{app:CAP},  
that Gilbert and Lynch's formulation can be recovered from ours.  

In Section~\ref{sec:logic}, we introduce a primitive logic of behaviours in which 
formulae are satisfied relative to `behaviours as worlds'. We consider the basic
classical propositional connectives and modal operators, so obtaining a syntactic 
variant of {\bf S5}. We establish the equivalence of behavioural equivalence and logical 
equivalence (see, for example, \cite{vBB94,Stirling01,bdRV01}). We discuss 
local reasoning principles (in the sense, for example, of \cite{OHearn2019}) for 
the composition of systems. 
Finally, in Section~\ref{sec:discussion}, 
we discuss a range of directions for future research, including \emph{time} (Section~\ref{subsec:time}) and \emph{tableaux proof systems} for the logic  (Section~\ref{subsec:tableaux}). 

\section{The basic theory} \label{sec:basic}

Any formal model of a distributed system, including system models constructed through the `distributed systems metaphor' mentioned above,  must capture the following two key points:
\begin{itemize}
    \item[-] \textbf{The topology of the system:} How do the components of the system sit together?
    \item[-] \textbf{The dynamics of the system:} How does the system as a whole and its components behave?
\end{itemize}

Of these two, modelling the topology is unproblematic: We can usually take some graph-like structure, where edges denote potential communication channels between the parts of a system.

Modelling the dynamics of a system is a more subtle issue. It is usually not feasible to give a full realistic description of a system's behaviour; therefore one works with abstractions. The choice of abstraction is then dependent on the intended application. Some common formalizations of behaviour include \emph{state sets}, (finite or infinite) \emph{traces} and possibly stochastic extensions thereof.

We seek a common abstraction of these formalizations. To this end, it is useful to contemplate the relation between the topological and the dynamical aspect of a distributed system. 
We can always observe the following \emph{reflection principle}:
\begin{quote}
The global behaviour of a system is reflected in the local behaviour of its components.
\end{quote}
This principle connects the topological and dynamical aspects of a system by expressing a functional relationship between global and local behaviours, manifesting in various ways. 
For example, the global behaviour of an input/output system might be modelled by the trace $\langle read(x),write(y),write(y'),read(x')\rangle$ of interactions; at the component responsible for reading the inputs, this global behaviour will be reflected as $\langle read(x),read(x')\rangle$. Cyberphysical systems might have complex behaviours, including a measuring of temperature, and one simple component that only records whether the system is overheated; then the global behaviour of the system is reflected by one of the behaviours $overheat=0$ or $overheat=1$.

The above reflection principle forms the core of our modelling approach. It then turns out that we can leave the question what behaviours `really are' largely open. What matters is only the relation between global and local behaviours, as mediated by the reflection function.

\subsection{Systems and Implementations}

We assume a given collection $\C$ of \emph{basic components} that will form the building blocks of our distributed systems. Each $c\in \C$ comes equipped with a set $\Beh{c}$, whose elements we call \emph{behaviours}. For nonempty $C\subseteq\C$ we call $\Beh{C}:=\Pi_{c\in C}\Beh{c}$ a \emph{set of $C$-snapshots}.

\begin{definition}[system]
\label{def:system}\mbox{}
\begin{itemize}
    \item[-] A \emph{system} is a function whose codomain is a set of snapshots.
    \item[-] If $f:B\to \Beh{C}$ is a system, we call the elements of $\Beh{f}:=B$ and $\Comp{f}:=C$ the \emph{behaviours of $f$} and the \emph{components of $f$},  respectively.
\end{itemize}
\end{definition}

We think of a system as a set $B$ of high-level descriptions of system behaviours together with a translation $f$ into concrete component behaviours. In this way the definition of a system proposed here is not a fully semantic definition: We do allow, e.g., that $f(b_1)=f(b_2)$ for \emph{different} $b_1, b_2\in B$. Such behaviours then differ only by their high-level description, not by any observable factors component behaviour. The motiviation for allowing syntactic elements is that it facilitates the modelling of a system.

On the other hand, a fully semantic view of systems could be obtained by requiring $f$ to be injective. In this case, a system can be identified with a (labelled) subset of $\Beh{C}$.

\begin{example}
\label{ex:first}
   Consider two components $c,d\in \C$ with $\Beh{c}=\{x\}^*$ and $\Beh{d}=\{y,z\}^*$. Then every $S\subseteq \{x,y,z\}^*$ can be seen as the behaviours of a system $f$ with $\Comp{f}=\{c,d\}$ where $f$ projects on substrings, e.g. $f(xxyxzzxy)=(xxxx,yzzy)$. Note that this $f$ is not injective, as we have $f(xxyxzzxy)=f(xxxxyzzy)$.
\end{example}

Before coming to our second crucial definition, we introduce some notational conventions. If $f:B\ra \Beh{C}$ and $D\subseteq C$, we denote $f_D:B\ra\Beh{D}$ the composition of $f$ with the obvious projection function $\Beh{C}\ra\Beh{D}$. If $D$ is a singleton $\{d\}$ we  write $f_d$ instead of $f_{\{d\}}$. In Example~\ref{ex:first} we have $f_d(xxyxzzxy)=yzzy$. It will be useful to view sets of components as systems; for this, we identify $C\subseteq \C$ with the identity function $\Beh{C}\ra\Beh{C}$.

\begin{definition}[implementation]
\label{def:implementation}
Let $f,g$ be systems with $\Comp{g}\subseteq\Comp{f}$. An \emph{implementation of $g$ in $f$} is a function $\sigma:\Beh{f}\ra\Beh{g}$ such that $f_{\Comp{g}}=g\circ\sigma$.
\end{definition}

In other words, $\sigma:\Beh{f}\ra\Beh{g}$ is an implementation if the following diagram commutes:
\begin{center}
    \begin{tikzcd}[ampersand replacement=\&]
	\& {\Beh{f}} \\
	{\Beh{g}} \\
	\& {\Beh{E}}
	\arrow["\sigma"', from=1-2, to=2-1]
	\arrow["g"', from=2-1, to=3-2]
	\arrow["f_E", from=1-2, to=3-2]
\end{tikzcd}
$E:=\Comp{g}$
\end{center}
In this case we write $f\impl[\sigma] g$ and say that \emph{$f$ implements $g$ (via $\sigma$)}, or that $g$ is a \emph{subsystem} of $f$. If $f$ implements $g$, then the local behaviour of $f$ at $c\in\Comp{g}$ can always be explained in terms of the behaviour of $g$. We say that $\sigma$ \emph{realizes $y\in\Beh{g}$} if there is an $x\in\Beh{f}$ such that $\sigma(x)=y$. 

As a simple example of Definition~\ref{def:implementation}, systems implement their components: For $C\subseteq\Comp{f}$ we have $f\impl[f_C] C$ because $f_C=id_C\circ f_C$. Here we have used the identification of $C$ with the system $id_C$ as described above.

It should be stressed that $f\impl[\sigma] g$ does \emph{not} mean that $f$ is in control of $g$. While $\sigma$ denotes a relation between the behaviours of $f$ and $g$, this relation is not necessarily causation. In fact, our model does not have an explicit notion of causation, say via input and output polarities (inputs react, outputs act). On the other hand, such notions can be derived from the set-up, for example:

\begin{definition}\label{def:input}
    If $f\impl[\sigma] g$ and $\sigma$ is \emph{surjective}, then $\sigma$ is called an \emph{input implementation}. $E\subseteq\Comp{f}$ is an \emph{input of $f$} if $f_E$ is an input implementation.
\end{definition}

The intuition is the following: if $\sigma$ is surjective, then every behaviour of $g$ is realized by $\sigma$. Thus $g$ is not constrained as a subsystem of $f$, and we can think of $f$ as merely `observing' the behaviour of $g$ and possibly reacting to it.

We end this section with a simple yet fundamental example that illustrates how a basic communication channel between two systems can be modelled in our framework. A general treatment of the composition of systems follows in the next section.

\begin{example}[message passing]
\label{ex:message}
Let $\Sigma$ be an alphabet, and let $c$ be a `communication component' with $\Beh{c}=\Sigma^*$; that is, behaviours are finite $\Sigma$-strings. Now for some arbitrary set $G\neq\emptyset$ of additional behaviours, let $g$ be a system with $\Beh{g}\subseteq G\times \{send(s)\mid s\in \Sigma^*\}$ that implements $c$ via $g_c(g,send(s))=s$. Furthermore let $f$ be another system with $\Beh{f}\subseteq F\times \{receive(s)\mid s\in\Sigma^*\}$ that implements $c$ via $f_c(f,receive(s))=s$ and where $F\cap G=\emptyset$, and assume moreover that $c$ is an input of $f$: for all $s\in\Sigma^\ast$, there exists $x\in f$ such that $(x,receive(s))\in\Beh{f}$.

We now construct a system $h$ that allows message passing from $g$ to $f$. For this let $\Comp{h}=\Comp{f}\cup\Comp{g}$ and let $\Beh{h}$ contain those triples $( y,transmit(s),x)$ such that $(y,send(s)) \in\Beh{g}$ and $(x,reveice(s))\in\Beh{f}$. Then $h$ implements both $f$ and $g$, and it stipulates that the messages that $g$ sends are those that $f$ receives. 
\end{example}

We conclude this section by introduction a notion of \emph{system equivalence}.\begin{definition}\label{def:equivalence}
    $f$ and $g$ are equivalent, written $f\equiv g$, if $f\impl g$ and $g\impl f$.
\end{definition}







\subsection{Compositions}




\begin{definition}[environment and composition]
\label{def:composition}
    An \emph{environment for for $g_1$ and $g_2$} is any system $f$ implementing both $g_1$ and $g_2$. An environment $f$ is a \emph{composition} if $\Comp{f}=\Comp{g_1}\cup\Comp{g_2}$.
\end{definition}

We 
write $g_1 \!\coimpl[\sigma_1] \! f \!\impl[\sigma_2] \! g_2$ to denote a system $f$ implementing both $g_1$ and $g_2$.

\begin{definition}
[interface]
\label{def:interface}
    $\Int{g_1}{g_2}:=\Comp{g_1}\cap\Comp{g_2}$ is called the \emph{interface} of $g_1$ and $g_2$. 
\end{definition}

If $g_1\coimpl[\sigma_1] f\impl[\sigma_2] g_2$ and $I:=\Int{g_1}{g_2}$ we have the following commutative diagram, assuming $I\neq\emptyset$:

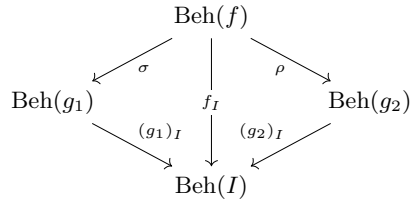
\begin{center}{\small
\begin{tikzcd}[ampersand replacement=\&]
	\& {\Beh{f}} \\
	{\Beh{g_1}} \&\& {\Beh{g_2}} \\
	\& {\Beh{I}}
	\arrow["\sigma", from=1-2, to=2-1]
	\arrow["\rho"', from=1-2, to=2-3]
	\arrow["{(g_1)_I}", from=2-1, to=3-2]
	\arrow["{(g_2)_I}"', from=2-3, to=3-2]
	\arrow["{f_I}"{description}, from=1-2, to=3-2]
\end{tikzcd}}
\end{center}


Call $(y_1,y_2)\in\Beh{g_1}\times\Beh{g_2}$ \emph{compatible} if $I=\emptyset$ or $(g_1)_{I}(y_1)=(g_2)_{I}(y_2)$. One is often interested in compositions which are `derived' from the systems $g_1$ and $g_2$; such compositions impose only those restrictions on its subsystems that are necessary to compose them. This is captured in the next definition.

\begin{definition}[free composition]
A composition $g_1\coimpl[\sigma_1] f\impl[\sigma_2] g_2$ is called \emph{free} if for every compatible pair $(y_1,y_2)\in\Beh{g_1}\times\Beh{g_2}$ there exists $x\in\Beh{f}$ such that $\sigma_1(x)=y_1$ and $\sigma_2(x)=y_2$.
 \end{definition}

For example, the system $h$ from Example~\ref{ex:message} is a free composition.
Given two systems $f,g$, there is a canonical free composition of $f$ and $g$ whose behaviours are exactly the compatible pairs and which has projections as implementation functions. We denote this composition by $f\comp g$; it is essentially the category-theoretic \emph{pullback} of the maps $(g_1)_I$ and $(g_2)_I$.\footnote{Pullbacks as models of synchronization also appear in \cite{pullback}.} 

\begin{definition}[canonical free composition]
\label{def:freecomp}
    Let $f,g$ be two systems. Let $I:=\Int{f}{g}$. $f\comp g$ is the composition of $f$ and $g$ whose behaviour set is $\{(x,y)\in\Beh{f}\times\Beh{g}\mid I=\emptyset\lor f_I(x)=g_I(y)\}$,  whose function is given by
    \[
    (f\comp g)_c(x,y)=
    \begin{cases}
    f_c(x)\quad\text{if }c\in\Comp{f} \\
    g_c(y)\quad\text{if }c\in\Comp{g}\setminus\Comp{f}
    \end{cases}
    \]
\end{definition}
Note that if $c\in\Int{f}{g}$, then $(f\comp g)_c(x,y)=f_c(x)=g_c(y)$ by compatibility.

We say that $f$ is \emph{runnable} if $\Beh{f}\neq\emptyset$ (or equivalently, $f\neq\emptyset$). The existence of at least one compatible pair is the minimal requirement for runnable compositions:

\begin{lemma}
\label{lem:composable}
If $\Int{f}{g}\neq\emptyset$, a runnable composition of $f$ and $g$ exists iff there is a compatible pair.
\end{lemma}
\begin{proof}
If some compatible pair exists, then $f\comp g$ as described in Def.~\ref{def:freecomp} above is runnable. Conversely, let $f\coimpl[\sigma]h\impl[\rho] g$ be any runnable composition of $f,g$ and pick $z\in\Beh{h}$. Let $x=\sigma(z)$, $y=\rho(z)$, and $I=\Int{f}{g}$. Then, by the properties of implementations, we can show that $(x,y)$ is a compatible pair: $f_I(x)=f_I(\sigma(z))=h_I(z)=g_I(\rho(z))=g_I(y)$.
\end{proof}

We can show that `directed' compositions are always runnable:

\begin{lemma}
Assume that $g$ is runnable. If $\Int{f}{g}=\Comp{e}$ and the system $e$ is an input of $f$ (see Def.~\ref{def:input}), then there is a runnable composition of $f$ and $g$. 
\end{lemma}
\begin{proof}
Let $f\impl[\rho]e$ be an input implementation and $g\impl[\sigma]e$ any implementation. Since $g$ is runnable, there exists $y\in\Beh{g}$, and as $e$ is an input component we can find $x\in\Beh{f}$ such that $\rho(x)=\sigma(y)$. For $I:=\Comp{e}$ we consequently have $f_I(x)=e(\rho(x))=e(\sigma(y))=g_I(y)$ and so $(x,y)$ is a compatible pair. Thus $f$ and $g$ are composable by Lemma~\ref{lem:composable}.
\end{proof}

\begin{example}
\label{ex:traces}
    Let $f,g$ be systems with $\Beh{f}\subseteq \{a,b,c\}^*$ and $\Beh{g}\subseteq \{b,c,d\}^*$, both closed under inital segments (think traces of actions), and let $\Int{f}{g}=\{e\}$ where $\Beh{e}=\{b,c\}^*$. Assuming that $f$ and $g$ implement $e$ via the obvious projections, $\Beh{f\comp g}$ consists of pairs $(s_1,s_2)\in\{a,b,c\}^*\times\{b,c,d\}^*$ whose restriction to the common actions $\{b,c\}$ coincides. For example, if $x=\langle a,b,b,c,a,b\rangle\in\Beh{f}$ and $y=\langle b,d,b,c,d,b\rangle\in\Beh{g}$ then $(x,y)\in\Beh{f\comp g}$, as their restriction to $\{b,c\}$ coincides in the trace $\langle b,b,c,b\rangle$.
    We can interpret this as a partially asynchronuous computation: 
\end{example}
    \begin{center}{\small 
        \begin{tikzcd}[ampersand replacement=\&]
	{f:} \& a \& b \&\& b \& c \& a \& b \\
	{g:} \&\& b \& d \& b \& c \& d \& b
	\arrow[Rightarrow, no head, from=1-3, to=2-3]
	\arrow[Rightarrow, no head, from=1-5, to=2-5]
	\arrow[Rightarrow, no head, from=1-6, to=2-6]
	\arrow[Rightarrow, no head, from=1-8, to=2-8]
	\arrow[from=1-2, to=1-3]
	\arrow[from=1-3, to=1-5]
	\arrow[from=1-5, to=1-6]
	\arrow[from=1-6, to=1-7]
	\arrow[from=1-7, to=1-8]
	\arrow[from=2-3, to=2-4]
	\arrow[from=2-4, to=2-5]
	\arrow[from=2-5, to=2-6]
	\arrow[from=2-6, to=2-7]
	\arrow[from=2-7, to=2-8]
\end{tikzcd}}
    \end{center}

\begin{example}[heaps]
Take as basic components a set $L$ of \emph{locations}, with $\Beh{l}=\mathbb{N}$ for $l\in L$.  For $L_0\subseteq^{fin} L$, $S(L_0)$ is the system whose behaviours are the $L_0$-based \emph{heaplets}, i.e.~functions $h:L_0\ra \mathbb{N}$. $S_0$ implements every $l\in L_0$ via the function $h\mapsto h(l)$. An $L_0$-stream is the system whose behaviours are sequences $\sigma$ of commands $l:=n$ where $l\in L_0$ and $n\in\mathbb{N}$. It implements $S_0$ via the function that sends $\sigma$ to the the heaplet $h$ where $h(l)$ is the value $n$ of the last command $l:=n$ in $\sigma$. Let $L_0$ be the disjoint union of $L_1$ and $L_2$. Then every $L_0$-stream is the free composition of an $L_1$-stream and an $L_2$-stream.
\end{example}

In Example~\ref{ex:traces}, nothing is said about the temporal relation between actions outside the interface, and so we have a model of `true' asynchronicity. On the other hand, if $f\coimpl[\sigma] h\impl[\rho] g$ is any free composition, then one can show that $h$ factors through $\Beh{f\comp g}$ via a \emph{surjective} map:

{\small \begin{center}
\begin{tikzcd}[ampersand replacement=\&]
	\& {\Beh{h}} \\
	\& {\Beh{f\comp g}} \\
	{\Beh{f}} \&\& {\Beh{g}} \\
	\& {\Beh{I}}
	\arrow["f_I"', from=3-1, to=4-2]
	\arrow["g_I", from=3-3, to=4-2]
	\arrow[two heads, from=1-2, to=2-2]
	\arrow["(f\comp g)_I", from=2-2, to=4-2]
	\arrow[from=2-2, to=3-1]
	\arrow[from=2-2, to=3-3]
	\arrow["\sigma", from=1-2, to=3-1]
	\arrow["\pi", from=1-2, to=3-3]
\end{tikzcd}
\end{center}}

 Intuitively, this means that in any free composition $h$, $\Beh{h}$ extends $\Beh{f\comp g}$ with some extra information. This can be seen as the information about how a pair $(x,y)$ of compatible behaviours is combined together: is it sequential,  synchronous, or asynchronous?

\begin{example}
    Let $f$ and $g$ be as in Example~\ref{ex:traces}. Let $h$ be a system whose behaviours are all elements of $\{a,b,c,d\}^*$ whose restriction to actions in $\{a,b,c\}$ and $\{b,c,d\}$ lies in $\Beh{f}$ and $\Beh{g}$, respectively. $h$ is a free composition of $f$ and $g$ via the obvious projection maps. Assuming again that $x=\langle a,b,b,c,a,b\rangle\in\Beh{f}$ and that $y=\langle b,d,b,c,d,b\rangle\in\Beh{g}$ we now have, for example, that $\langle a,b,d,b,c,a,d,b\rangle\in\Beh{h}$. Thus $h$ `linearizes' the behaviours $x$ and $y$.
\end{example}

\section{Implementation Guarantees and the CAP theorem}
\label{sec:CAP}

The aim of this section is to prove a formal version of the Brewer's CAP theorem. Recall that the CAP theorem states that the three guarantees of \emph{consistency}, \emph{availability} and \emph{partition-tolerance} cannot be simultanously satisfied in a distributed system (cf.~the introduction to this article).

We start by formalizing the notion of an \emph{implementation guarantee}:

\begin{definition}[implementation guarantee]
An \emph{implementation guarantee} for a system $g$ is a set $\X\subseteq\mathcal{P}(\Beh{g})$. 
The implementation $f\impl[\sigma]g$ satisfies $\X$, written $\sigma\models \X$, if $\{\sigma(x)\mid x\in\Beh{f}\}\in \X$.
\end{definition}

\begin{definition}[consistency]
    For 
    $C\subseteq\Beh{g}$, the implementation guarantee
    $\{X\subseteq\Beh{g}\mid X\subseteq C\}$
     is called \emph{$C$-consistency guarantee}.
\end{definition}
Simply put, an implementation guarantee satisfies $C$-consistency if it only realizes behaviours in $C$. To avoid trivial solutions to consistency, one also needs to consider guarantees that sufficiently many behaviours are realized. We formula such \emph{availability guarantees} relative to a binary relation $R\subseteq\Beh{g}^2$. Henceforth $R(x,\bullet):=\{y\mid xRy\}$ and $dom(R):=\{x\mid R(x,\bullet)\neq\emptyset\}$.

\begin{definition}[availability]
    For a binary relation $R$ on $Beh(g)$, the implementation guarantees
    \begin{align*}
     \{X\subseteq \Beh{g}\mid\forall x\in X(x\in\dom(R)\text{ implies }\exists y\in X\cap R(x,\bullet)\}
     \\
     \{X\subseteq \Beh{g}\mid\forall x\in X (R(x,\bullet)\subseteq X)\}
    \end{align*}
    are called \emph{weak} and \emph{strong $R$-availability},  respectively.
\end{definition}
More directly, if an implementation realizes some behaviour $x\in \dom(R)$, then weak (resp.\ strong) $R$-availability states that it must also realize some (resp.\ all)  $y\in R(x,\bullet)$. Strong availability corresponds to the \emph{input-enabledness} condition used in I/O-automata~\cite{lynch1988}.

\begin{definition}[partition-tolerance]
    Let $f\impl[\sigma_1]g_1$ and $f\impl[\sigma_2]g_2$. The implementation guarantee
    \[
    \{
    X\subseteq\Beh{f}\mid 
    \forall x_1,x_2\in X \exists x\in X \sigma_1(x)=\sigma_1(x_1),\sigma_2(x)=\sigma_2(x_2)
    \}
    \]
    is called \emph{$(\sigma_1,\sigma_2)$-partition-tolerance}.
\end{definition}

Partition-tolerance states that every pair $\sigma_1(x_1)$ and $\sigma_2(x_2)$ of realized behaviours of $g_1$ and $g_2$, where possibly $x_1 \neq x_2$, can also be \emph{jointly} realized by some single behaviour $x$. In other words, implemented behaviours in $g_1$ never restrict implemented behaviours in $g_2$ and vice versa. This means that no communication (that is, `message passing') between the components $g_1$ and $g_2$ is necessary to run the system (neither has to wait for the other).

The final definition needed to state the CAP theorem is that of entanglement. We think here of a situation where a system has two subsystems, comes with a consistency guarantee, and there are two relations modelling different ways of evolving the system. Entanglement then denotes that, for every initial behaviour and evolution to an intermediate behaviour via the first relation, there is a further evolution into a `problematic' behaviour via the second relation such that the following holds: 
\begin{quote}
Whenever two behaviours look like the problematic behaviour from the perspective of the first subsystem, but look like the behaviour of the initial and intermediate behaviour from the perspective of the second subsystem, then at at least one of them fails consistency.
\end{quote}

\begin{definition}[entanglement]
Let $f\impl[\sigma_1]g_1$ and $f\impl[\sigma_2]g_2$ and $C\subseteq\Beh{f}$. Two relations $R,S\subseteq\Beh{f}^2$ of behaviours are \emph{$C$-entangled over $(\sigma_1,\sigma_2)$} if they satisfy the condition
    \begin{align*}
    \label{eq:incomp}
    \forall w\in\Beh{f}\exists x\in R(w,\bullet)\cap dom(S)\forall v\in S(x,\bullet)\forall y,z\in\Beh{f}:\\        \sigma_1(y)=\sigma_1(z)=\sigma_1(v)\land\sigma_2(y)=\sigma_2(w)\land\sigma_2(z)=\sigma_2(x) \\
        \Rightarrow y\notin C \lor z\notin C
        \tag{$\ast$}
    \end{align*}
which is illustrated by the following picture:
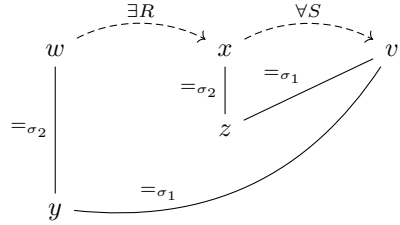
\begin{center}
\begin{tikzcd}[ampersand replacement=\&]
	w \&\& x \&\& v \\
	\&\& z \\
	y
	\arrow["{\exists R}", curve={height=-12pt}, dashed, from=1-1, to=1-3]
	\arrow["{\forall S}", curve={height=-12pt}, dashed, from=1-3, to=1-5]
	\arrow["{=_{\sigma_2}}"', no head, from=1-3, to=2-3]
	\arrow["{=_{\sigma_1}}", no head, from=2-3, to=1-5]
	\arrow["{=_{\sigma_1}}"{pos=0.3}, curve={height=30pt}, no head, from=3-1, to=1-5]
	\arrow["{=_{\sigma_2}}", no head, from=3-1, to=1-1]
\end{tikzcd} 
\end{center}
\end{definition}

We now state and prove a formalized version of the CAP theorem.

\begin{theorem}[generalized CAP theorem]
\label{thm:CAP}
    Let $f\impl[\sigma_1]g_1$ and $f\impl[\sigma_2]g_2$. Let $C\subseteq\Beh{f}$ and assume that $R,S\subseteq\Beh{f}^2$ are $C$-entangled over $(\sigma_1,\sigma_2)$.

    \medskip 
    \noindent Then every runnable implementation of $f$ must violate at least one of the following guarantees:
   \begin{enumerate}
       \item[-] $C$-consistency
       \item[-] strong $R$-availability and weak $S$-availability
       \item[-] $(\sigma_1,\sigma_2)$-partition-tolerance
   \end{enumerate}
\end{theorem}
\begin{proof}
    Let $h\impl[\sigma] f$ be any implementation of $f$, and assume that it satisfies (2) and (3). We show that $\sigma \nvDash\text{$C$-consistency}$.  
    Let $w\in\Beh{f}$ be any behaviour realized by $\sigma$, and pick $x$ as in the entanglement condition. By strong $R$-availability, $\sigma$ realizes $x$. By weak $S$-availability and the fact that $x\in dom(S)$, $\sigma$ also realizes some $v\in S(x,\bullet)$. By $(\sigma_1,\sigma_2)$-partion-tolerance $\sigma$ also realizes behaviours $y,z$ such that $\sigma_1(y)=\sigma_1(z)=\sigma_1(v)$ and $\sigma_2(y)=\sigma_2(w)$, $\sigma_2(z)=\sigma_2(x)$. But then by entanglement, 
    $\sigma$ realizes a behaviour that is not in $C$, violating $C$-consistency.
\end{proof}


Of course, the usual formulation of the CAP theorem (cf.~\cite{Lynch2012}) does not involve a complicated entanglement condition. This is because it deals with a concretely given set of implementation guarantees. Our formulation on the other hand starts from arbitrary guarantees derived from the sets $C$, $R$, and $S$, none of which needs to be a `read' or `write' guarantee as in~\cite{Lynch2012}. The entanglement condition is a logical abstraction of the relation between read and write guarantees in CAP that allows the proof to go through. 
To verify that Theorem~\ref{thm:CAP} is indeed a generalization of CAP, we show in Appendix~\ref{app:CAP} how to recover a standard form of the CAP theorem from it. This also shows in detail how a more traditional process-based view of systems can be modelled in our behaviour-based view.

\section{Logic and local reasoning principles} \label{sec:logic}

We now describe a simple logic of systems, and use this logic to state some local reasoning principles. 

Our logic is one of system behaviours, asserting properties of behaviours through a satisfaction relation between systems and formulae of the form 
\[
    f, \val , x \models \phi 
\]
where $f$ is a system, $x$ is a behaviour of $f$, and 
$\mathcal{V}$ is a valuation. This relation will be defined formally below. 

The idea of such a logic is related to logics 
of state in the sense of van Benthem-Bergstra and Hennessy-Milner (commonly known as `Hennessy-Milner' theorems) --- see, for example, \cite{vBB94,Milner1989,Stirling01} --- in which a process with state $E$ is judged to satisfy a property $\phi$, denoted $E \models \phi$. In this setting, `action modalities' such as $[a]\phi$, where $a$ denotes a basic action, play a key role in connecting the processes and the logic: 
\[
\begin{array}{rcl}
  E \models [a]\phi & \mbox{iff} & 
    \mbox{for all evolutions $E \goes{a} E'$, 
        $E' \models \phi$} 
\end{array}
\]
The key result in such a setting is that bisimulation equivalence of processes is equivalent to their logical equivalence. 

In our setting, we establish what may be understood as a corresponding result between systems with equivalent behaviours and their logical properties. Our result does not, however, require the presence of action modalities. Behaviours are a basic component of models and the formula $\Box \phi$ simply 
characterizes all of the behaviours of the system. This is explained formally in 
Section~\ref{subsec:universal}.

We employ also structural connectives to characterize 
the compositional structure of systems and, through a key concept of 
\emph{interface}, how the behaviour of one component of a system affects a component to which it is connected. We introduce `local reasoning' principles, borrowing the name from Separation Logic \cite{IO01,Reynolds2002,OHearn2019}, to characterize these 
ideas. Connectives of this kind are also present in logics of state for processes; see, for example, 
\cite{Dam89,CP2009,CMP2012,AP16}).

\subsection{The elementary logic}

For every basic component $c\in\C$, we assume a set $Var(c)$ of variables
. All of these sets are assumed to be disjoint. A \emph{valuation} is a function $\val$ mapping each $p\in Var(c)$ to a subset of $\Beh{c}$. 

For $C\subseteq\C$, the \emph{elementary $C$-logic} $\mathcal{L}(C)$ is built from 
variables in $\bigcup_{c\in C} V(c)$ and the boolean connectives $\land$ and $\lnot$. It is interpreted over triples $(f,\mathcal{V},x)$ where $f$ is a system with $\Comp{f}\supseteq C$ and $x\in\Beh{f}$ using the following clauses:
\begin{align*}
    f,\val,x\models p&:\iff f_c(x)\in \val(p) \qquad\qquad(p\in Var(c))\\
    f,\val,x\models \varphi\land\psi&:\iff f,\val,x\models \varphi\text{ and }f,\val,x\models\psi \\
    f,\val,x\models \lnot\varphi &:\iff f,\val,x \nvDash \varphi 
\end{align*}



We write $f,\val\models\varphi$ if $f,\val,x\models\varphi$ for all $x\in\Beh{f}$. As usual, we can define further boolean connectives $\varphi\lor\psi:=\lnot(\lnot\varphi\land\lnot\psi)$ and $\varphi\to\psi:=\lnot\varphi\lor\psi$. We let $\mathcal{L}:=\mathcal{L}(\C)$ and $\mathcal{L}(f):=\mathcal{L}(\Comp{f})$.

For example, for $p\in Var(c)$ and $q\in Var(d)$ we have $f,\val\models p\to q$ iff $f_d(x)\in\val(q)$ whenever $f_c(x)\in\val(p)$. That is, if some behaviour of $f$ locally satisfies $p$ at the component $c$, then it locally satisfies $q$ at $d$.

%
%
%
%
%

\begin{lemma}[absoluteness]
    \label{lem:transfer}
    Let $f\!\impl[\sigma]\!g$ and $\alpha\!\in\!\mathcal{L}(g)$. Then, for all $x\in\Beh{f}$,
    \[
    f,\val,x\models \alpha\iff g,\val,\sigma(x)\models\alpha
    \]
\end{lemma}
\begin{proof}
    By induction on $\alpha$. We only show the case $\alpha=p_c$, which uses the properties of implementations:
    \[
    f,\val,x\models p_c\iff f_c(x)\in \val(p_c)\iff g_c(\sigma(x))\in \val(p_c)\iff g,\val,\sigma(x)\models p_c
    \]
    \end{proof}

As a corollary, we can observe that $g,\val\models\alpha$ implies $f,\val\models\alpha$ whenever $f\impl[\sigma] g$ and $\alpha\in\mathcal{L}(g)$. Thus properties of $g$ that are expressible in $\mathcal{L}(g)$ are inherited by any system implementing $g$.

We can state a first simple local reasoning principle: In order to prove that a system satisfies some $\mathcal{L}$-property $\beta$, it suffices to find subsystems $g$ and $h$ of $f$ and a property $\alpha$ such that $g$ enforces that $\alpha$ is true at the interface of $f$ and $g$, and $h$ satisfies $\beta$ whenever it satisfies $\alpha$ at its interface.

\begin{theorem}[Local reasoning I]
\label{thm:lrI}
    Let $g\coimpl[\sigma]f\impl[\pi]h$ and $I=\Int{f}{g}\neq\emptyset$. The rule
    \[
    \infer{f,\val\models \beta}
        {
        g,\val\models\alpha
        &
        h,\val\models\alpha\ra\beta
        }
    \]
    is admissible for all $\alpha\in\mathcal{L}(I)$ and $\beta\in\mathcal{L}(h)$.
\end{theorem}
\begin{proof}
    Let $x\in\Beh{f}$. By assumption, $g,\val,\sigma(x)\models \alpha$. As $g$ implements the interface $I$ (seen as a system), Lemma~\ref{lem:transfer} implies $I,\val,g_I(\sigma(x))\models \alpha$. Furthermore,  $g_I(\sigma(x))=h_I(\pi(x))$, so we have $I,\val,h_I(\pi(x))\models\alpha$. Another application of Lemma~\ref{lem:transfer} yields $h,\val,\pi(x)\models\alpha$. By assumption we have $h,\val,\pi(x)\models \alpha\ra\beta$ and so $h,\val,\pi(x)\models\beta$ follows. Thus, by yet another application of Lemma~\ref{lem:transfer}, $f,\val,x\models\beta$ as desired.
\end{proof}

Note that this theorem applies in particular to $f=g\comp h$ (see Def.~\ref{def:freecomp}).

\begin{example}
    Consider two components $c$ and $d$ with $\Beh{c}=\Beh{d}=\mathbb{Q}$. Let $g$ be a system implementing $c$ and having $d\notin\Comp{g}$. We think of $g$ as a system that, among other things, records the temperature in Celsius in its component $c$. Moreover let $h$ be a converter from Celsius to Fahrenheit, modelled as follows: $\Comp{h}=\{c,d\}$, $h=id:\Beh{h}\to\Beh{c}\times\Beh{d}$ where $\Beh{h}=\{(x,y)\in\mathbb{Q}\mid y=x*1.8+32\}$. Let $p_c\in V(c)$ and $p_d\in V(d)$ be variables with $\val(p_c)=[0,10]$ and $\val(p_d)=[32,50]$. Assume we know that $g\models p_c$, that is, only temperatures between $0$ and $10$ degree Celsius are measured by $g$. Then we can formally derive that $g\comp h$ only records temperatures between $32^\circ$ and $50^\circ$ Fahrenheit:
    \[
    \infer{g\comp h,\val\models p_d}
        {
        g,\val\models p_c 
        &
        h,\val\models p_c \to p_d
        }
    \]
\end{example}

\subsection{Adding a universal modality}
\label{subsec:universal}

A standard extension $\mathcal{L}(C)\subseteq\mathcal{L}(C)^\Box$ adds a universal modal operator $\Box$ to the language, subject to the following clause for satisfaction:
\[
f,\val,x\models\Box\varphi :\iff f,\val,y\models\varphi \text{ for all }y\in\Beh{f}
\]
$\Box \varphi$ therefore denotes a property $\varphi$ that holds for \emph{all} behaviours $x \in \Beh{f}$. This corresponds to the standard clause for $\Box$ in modal logic in the case that the accessibility relation is universal; that is, in which each element is related to every other element. We therefore obtain a syntactic variant of the well-known modal logic $\mathbf{S5}$. 

In the setting or process algebras, such as those discussed above that provide bases for simulation modelling, it is more natural to work with action modalities/formulae such as $[a]\phi$,  as discussed above, in which the action $a$ refers to an evolution of process state. In our setting, the simple modality refers directly to behaviours. In both cases, however, van~Benthem-Bergstra-Hennessy-Milner equivalence theorems can be obtained --- see Theorem~\ref{thm:HM}. 

In the presence of the modal operator $\Box$,  Lemma~\ref{lem:transfer} fails, but some variants of it remain provable. Call a formula in $\mathcal{L}^\Box$ \emph{positive} (resp.\ \emph{negative}) if no $\Box$ appears in the scope of an odd (resp.\ even) number of $\lnot$'s. By a simple induction we can show the following:

\begin{lemma}[directed absoluteness]
    \label{lem:dtransfer}
    Assume $f\impl[\sigma] g$. Then for all $x\in\Beh{f}$,
    \begin{align*}
     &g,\val,\sigma(x)\models\alpha\text{ implies }f,\val,x\models \alpha \quad\text{for positive }\alpha\in\mathcal{L}(g)^\Box  \\
     \text{and}\quad &f,\val,x\models \alpha \text{ implies }g,\val,\sigma(x)\models\alpha\quad\text{for negative }\alpha\in\mathcal{L}(g)^\Box.
    \end{align*}
\end{lemma}

By tracking the absoluteness properties used in the proof of Theorem~\ref{thm:lrI}, we can generalize 
slightly:

\begin{theorem}[Local reasoning II]
    Let $g\coimpl[\sigma]f\impl[\pi]h$ and $I=\Int{f}{g}\neq\emptyset$. The rule
    \[
    \infer{f,\val\models \beta}
        {
        g,\val\models\alpha
        &
        h,\val\models\alpha\ra\beta
        }
    \]
    is admissible for all $\alpha\in\mathcal{L}(I)$ and positive $\beta\in\mathcal{L}(h)^\Box$.
\end{theorem}

The next local reasoning principle makes use of the notion of an input (Definition~\ref{def:input}).

\begin{theorem}[Local reasoning III]
    Assume $I=\Int{f}{g}\neq\emptyset$ is an input to $g$. The rule
    \[
    \infer{f\comp g,\val\nvDash \beta}
        {
        f,\val\nvDash\alpha
        &
        g,\val\models\lnot\alpha\ra\lnot\beta
        }
    \]
    is admissible for all $\alpha\in\mathcal{L}(I)$ and negative $\beta\in\mathcal{L}(g)^\Box$.
\end{theorem}
\begin{proof}
    Pick some $x\in\Beh{f}$ such that $f,\val,x\models\lnot\alpha$. By absoluteness, we have $I,\val,f_I(x)\models\lnot\alpha$. Since $I$ is an input to $g$, we can find $y\in\Beh{g}$ such that $g_I(y)=f_I(x)$. Again by absoluteness we obtain $g,\val,y\models\lnot\alpha$ and hence $g,\val,y\models\lnot\beta$. As $\lnot\beta$ is positive it follows that $f\comp g,\val,( x,y)\models\lnot\beta$.
\end{proof}

Thus to show that some negative $\beta$ fails in $f\comp g$ where the interface is an input to $g$, it suffices to show that $f$ fails some property $\alpha$ at the interface and that this failure causes the failure of $\beta$ in $g$. This needs not be true if $f$ is not an input, as then the behaviour failing $\alpha$ might not be realized in $f\comp g$.


For systems with a trivial interface we can show the following:

\begin{theorem}
    The rule
    \[
    \infer{g\comp h,\val\models \beta}
        {
        h,\val\models\beta
        }
    \]
    is admissible for all $\beta\in\mathcal{L}(h)^\Box$ given that $\Int{g}{h}=\emptyset$.
\end{theorem}
\begin{proof}
    It suffices to prove by induction over $\delta\in\mathcal{L}(h)^\Box$ that 
    \[
    g\comp h,\val,(x,y)\models\delta\iff h,\val,y\models\delta
    \]
    This follows from the fact that $\Beh{g\comp h}=\Beh{g}\times\Beh{h}$ whenever $\Int{g}{h}=\emptyset$. We illustrate one step of the induction:
    \begin{align*}
        g\comp h,\val,(x,y)\models\Box\delta'&\iff \forall x'\in\Beh{g}\forall y'\in\Beh{h}.\,\, g\comp h,\val,(x',y')\models\delta' \\
        & \overset{IH}{\iff} \forall x'\in\Beh{g}\forall y'\in\Beh{h}.\,\,h,\val,y'\models\delta' \\
        & \iff \forall y'\in\Beh{h}.\,\, h,\val,y'\models\delta' \\
        & \iff h,\val,y\models\Box\delta'
    \end{align*}
    
\end{proof}

\noindent \textit{Remark}: If $\beta$ encodes a Hoare triple, then the above is a frame rule (cf. \cite{IO01}).

Finally, we state a Hennessy-Milner Theorem (see, e.g., Theorem~2.24 in \cite{bdRV01}) for our logic. The usual reference to bisimulations can be avoided here as we are working with the simple modal logic $\mathbf{S5}$; instead, we use the notion of a \emph{type}.

\begin{definition}
For a system $f$ with $\Comp{f}=C$ and $x\in\Beh{f}$ we define $\mathsf{tp}(x):=\{p\in V(C)\mid f,x\models p\}$. Moreover, we define $\mathsf{tps}(f):=\{\mathsf{tp}(x)\mid x\in\Beh{f}\}$.
\end{definition}

\begin{theorem}\label{thm:HM}
    Let $f,g$ be two systems with $\Comp{f}=\Comp{g}=:C$, and assume $V(C)$ is finite.\footnote{For simplicity we here use this condition stronger than \emph{image-finiteness} that appears in Hennessy-Milner theorems in modal logic, see, for example,  
    \cite{Stirling01}.} For any $(x,y)\in\Beh{f}\times\Beh{g}$, the following are equivalent:
    \begin{enumerate}
         \item $f,x\models\varphi\iff g,y\models\varphi$, for all $\varphi\in\mathcal{L}(C)$,
        and 
        \item $\mathsf{tp}(x)=\mathsf{tp}(y)$.
    \end{enumerate}
    The following are also equivalent:
    \begin{enumerate}
         \item $f,x\models\varphi\iff g,y\models\varphi$, for all $\varphi\in\mathcal{L}(C)^\Box$,
        and 
        \item $\mathsf{tp}(x)=\mathsf{tp}(y)$ and $\mathsf{tps}(f)=\mathsf{tps}(g)$.
    \end{enumerate}
    
\end{theorem}

\begin{proof}
We show only the second equivalence. For (1)$\Rightarrow$(2), we argue as follows: 
    Let $D:=\mathsf{tp}(x)$, and define $\varphi_D:= \bigwedge\{p\mid p\in D\}\land\bigwedge\{\lnot p\mid p\in V(C)\setminus D\}.$ Then $f,x\models \varphi_D$ and so $g,y\models \varphi_D$ by assumption, which in turn implies $\mathsf{tp}(y)=\mathsf{tp}(x)=D$. Now let $D\in\mathsf{tps}(f)$, say $D=\mathsf{tp}(x')$ for some $x'\in\Beh{f}$. Then $f,x\models\lnot\Box\lnot\varphi_D$, and so by assumption $g,y\models\lnot\Box\lnot\varphi_D$. It follows that for some $y'\in\Beh{g}$, $g,y'\models\varphi_D$ and so $\mathsf{tp}(y')=D$, which implies $D\in\mathsf{tps}(g)$. Hence $\mathsf{tps}(f)\subseteq\mathsf{tps}(g)$, and by a symmetric argument we get $\mathsf{tps}(g)\subseteq\mathsf{tps}(f)$.

    (2)$\Rightarrow$(1) is achieved by induction on $\varphi$; we only consider the case that $\varphi=\Box\psi$. Assume towards a contradiction that $f,x\models\Box\psi$ and $g,y\nvDash\Box\psi$. The latter means that $g,y'\nvDash\psi$ for some $y'\in\Beh{g}$. By assumption, there is an $x'\in\Beh{f}$ such that $\mathsf{tp}(x')=\mathsf{tp}(y')$. Then by the induction hypothesis, $f,x'\nvDash \psi$, contradicting $f,x\models\Box\psi$. Conversely, if $f,x\nvDash\Box\psi$, then $f,x'\nvDash\psi$ for some $x'\in\Beh{f}$. Pick $y'\in\Beh{g}$ with $\mathsf{tp}(y')=\mathsf{tp}(x')$; by the induction hypothesis,  $g,y'\nvDash\psi$ and so $g,y\nvDash\Box\psi$ follows.
\end{proof}

\subsection{Adding structural connectives} \label{subsec:struct-conn}

We have discussed extensively the sense in which systems consist of components, connected by interfaces. We have also established an elementary logic of system behaviours and, while we have considered some 
basic notions of local reasoning, we have not yet given a logical characterization of the compositional structure of systems. 

To this send, we now consider two `structural' connectives that are closely related to the multiplicative conjunction found in relevance and bunched logics. In particular, their use to characterize system structure 
resembles the use of such connectives in the family of `separation logics' and related ideas \cite{IO01,Reynolds2002,Docherty2019,pym2019resource}. 


We define the following ternary relation between system--behaviour pairs, which syntactically overloads the $\comp$-operator from Definition~\ref{def:freecomp}:

\[
\begin{array}{rcl}
    \mbox{`$(f,x)=(f_1,x_1)\comp(f_2,x_2)$'} & :\iff & 
    \mbox{$f\equiv f_1\comp f_2$ and $\forall i\in\{1,2\}$} \\
    & & \mbox{$\forall c\in\Comp{f_i}: f_c(x)=(f_i)_c(x_i)$} \\
    \mbox{`$(f,x)=(f_1,x_1)\dcomp(f_2,x_2)$'} & :\iff & \mbox{$(f,x)=(f_1,x_1)\comp(f_2,x_2)$ and} \\
    & & \mbox{$\Int{f_1}{f_2}$ is an input to $f_2$}
\end{array}
\]
The first clause uses the equivalence notion $\equiv$ from Definition~\ref{def:equivalence}.

We can now define a logic $\mathcal{L}^*(C)$ that extends $\mathcal{L}(C)$ with binary operators $\ast$ and $\dast$, which are interpreted as follows:
\[
\begin{array}{rcl}
f,\val,x\models \varphi\ast\psi & \iff & 
    \mbox{$(f,x) = (f_1,x_1) \comp (f_2,x_2)$ for some} \\
    & & \mbox{\rm $f_1,\val,x_1\models\varphi$ and $f_2,\val,x_2\models\psi$} \\ 
f,\val,x\models \varphi\dast\psi & \iff & 
    \mbox{\rm $(f,x) = (f_1,x_1) \dcomp (f_2,x_2)$ for some} \\ 
    & & \mbox{\rm $f_1,\val,x_1\models\varphi$ and $f_2,\val,x_2\models\psi$} 
\end{array}
\]
Evidently, $\ast$ and $\dast$ are multiplicative conjunctions. The 
first characterizes a simple (ordered) decomposition of a system into two components. The second also does this, but, in addition, identifies the interface that supports the output of the first component to be the input to the second.

\begin{theorem}[local-global]
\label{thm:globlocal}
    The following is a theorem for all positive $\varphi$ in $\mathcal{L}^\Box$:
    \[
    \varphi\ast\top\ra\varphi
    \]
\end{theorem}
That is, if a positive formula $\varphi$ holds in a part of the system, then it holds in the whole system.
\begin{proof}
    Assume $f,\val,x\models\varphi\ast\top$, and let $(f_1,x_1),(f_2,x_2)$ be system-behaviour pairs as in the satisfaction clause for $\ast$. Thus $f_1,\val,x_1\models\varphi$. As $f$ implements $f_1$ via an implementation $\sigma$ that satisfies $\sigma(x)=x_1$ and $\varphi$ is positive, it follows from Lemma~\ref{lem:dtransfer} that $f,\val,x\models\varphi$.
\end{proof}
\begin{theorem}[global-local]
    The following is a theorem for all negative $\varphi$ in $\mathcal{L}^\Box$:
    \[
    \varphi\ra(\psi\ast\psi'\ra (\psi\land\varphi) \ast \psi')
    \]
\end{theorem}
That is, if a negative formula $\varphi$ holds globally and some formula $\psi$ holds locally, then $\varphi$ holds at that location as well.
\begin{proof}
    Assume both $\varphi$ and $\psi\ast\psi'$ hold in a system. Clearly we must have at least $(\psi\land\varphi)\ast\psi'$ or $(\psi\land\lnot\varphi)\ast\psi'$ in the same system, so it suffices to show that $(\psi\land\lnot\varphi)\ast\psi'$ cannot be the case. Indeed, as $\lnot\varphi$ is positive we have $(\lnot\varphi)\ast\top\to\lnot\varphi$, or equivalently $\varphi\to\lnot(\lnot\varphi\ast\top)$, from Theorem~\ref{thm:globlocal}. As we assume $\varphi$ this implies $\lnot(\lnot\varphi\ast\top)$, which is stronger than the required $\lnot((\psi\land\lnot\varphi)\ast\psi')$.
\end{proof}

Just as in linear, bunched, and relevance logics,  the multiplicative conjunction $\ast$ gives rise to another connective $\wand$ called the `magic wand':
\begin{align*}
 f,\val,x\models \varphi\wand\psi \iff &f\otimes g,\val,(x,y)\models \psi  \\
 &\text{whenever }g,\val,y\models\varphi
\end{align*}
Using $\wand$, a system $f$ can reason about its place inside a larger structure $f\otimes g$. More specifically, $f,\val,x$ satisfies $\varphi\wand\psi$ if, whenever `plugged into' a system satisfying $\varphi$, the combined system satisfies $\psi$. Using $\dcomp$ instead of $\comp$ in the clause above defines a `directed magic wand', which then specifies properties of combined systems where $f$ acts as an input.

We mention one further variant of a multiplicative conjunction. Defining
\begin{align*}
     \mbox{`$(f,x)=(f_1,x_1)\discomp(f_2,x_2)$'} & :\iff & \mbox{$(f,x)=(f_1,x_1)\comp(f_2,x_2)$ and} \\
    & & \mbox{$\Int{f_1}{f_2}=\emptyset$}
\end{align*}
we introduce $\ast_d$ via the clause 
\begin{align*}
f,\val,x\models \varphi\disast\psi & \iff & 
    \mbox{\rm $(f,x) = (f_1,x_1) \discomp (f_2,x_2)$ for some} \\ 
    & & \mbox{\rm $f_1,x_1,\val\models\varphi$ and $f_2,x_2,\val\models\psi$} 
\end{align*}
This --- very restrictive --- conjunction corresponds exactly to the multiplicative conjunction of Separation Logic \cite{IO01,Reynolds2002}, where it is stipulated that the heap is split into disjoint parts. In our context, this means that the two subsystems $f_1$ and $f_2$ have no components in common.

We leave a more systematic study of these structural connectives to further work. In particular, we would like to explore extensions of the Hennessy-Milner theorem (Theorem~\ref{thm:HM}) for $\ast$ and $\dast$.

\section{Discussion and Conclusion} \label{sec:discussion}

We have established a minimalistic, behaviour-based framework for systems modelling. We have shown that it encompasses key ideas from systems theory and have given a logical analysis that supports reasoning about compositional structure. In particular, it can be shown that the framework can be used to give a semantics to the modelling tools based on the distributed systems metaphor, as discussed in Section~\ref{sec:introduction}, though the details of this are beyond our present scope. 
This includes the approach to interfaces, composition, and local reasoning as described, for example, in \cite{CP15,caulfield2016mod,CIP2022}.

We have also shown that even a rather concrete property of distributed system --- namely the CAP theorem --- is witnessed in the framework.

In Section~\ref{sec:logic}, we have laid out a minimal logical theory that, in addition to having the standard classical connectives, is sufficient to characterize basic system evolution (through a modality) 
and compositional structure (through multiplicative conjunctions). 
The latter come in two versions, $\ast$ and $\stackrel{\rightarrow}{\ast}$, conveying whether a system can be split into subsystems satisfying certain properties, or extended to a larger system. These structural connectives, which are similar to those considered in logics for process algebras (for example, \cite{Dam89,CP09,CMP10,AP16}) and Separation Logic (for example, \cite{IO01,OHearn2019}).

\subsection{Time} \label{subsec:time}

In this section, we briefly sketch how to introduce an explicit notion of time into our set-up. There are several ways of achieving this. The obvious one would be to abandon the principle that behaviours are `unstructured' sets, instead imposing that they carrier a partial order $<$ which is interpreted as a `happens before' relation. In the trace model, we could then state for example that $\langle a,b\rangle<\langle a,b,c\rangle$; that is, the behaviour $\langle a,b,c\rangle$ is a \emph{continuation} of the behaviour $\langle a,b\rangle$.

It turns out, by adopting a different philosophy, that such a modification of the framework is not necessary. Instead of treating time as something inherent to a system, we think of it as being \emph{observed} from outside the system. Thus in order to introduce time to a system $f$, we should look at a larger system that includes an \emph{observer} of $f$.

\begin{definition}
    A \emph{timed implementation for $f$} is a composition of $f$ with some basic component $c$, where $f$ is an input and $c\notin\Comp{f}$ is a basic component whose behaviour set is a partial order. 
\end{definition}

 Here we think of $\Beh{c}$ as some form of bookkeeping of the behaviours of $f$: For $t,t'\in\Beh{c}$, $t< t'$ means that $c$ may experience $t$ before $t'$. Incompatible behaviours denote observations in different histories of the observer.

Let the timed implementation be given as $f\coimpl[\sigma]h\impl[\rho]c$. We can then import this ordering of events into the behaviours $x,y$ of $f$ by setting $x\leq y$ if whenever $\sigma(v)=x$ and $\sigma(w)=y$ for some behaviours $v,w\in\Beh{h}$, then we have $\rho(v)\leq\rho(w)$. Thus $x\leq y$ holds if $y$ can only be observed by $c$ after $x$. $\leq$-minimal elements in $\Beh{f}$ can be identified with initial states of $f$.

\subsection{Tableaux systems} \label{subsec:tableaux}

A substantial and theoretically important development --- beyond our present scope --- would be to provide a proof system for the logic, together with associated meta-theoretical results. One approach would be to develop a labelled tableaux system, similar to those developed for other logics involving mixes of classical and substructural connectives and modalities (see, e.g., \cite{GKP2020,Larchey-Wendling2014,pym2019resource}). 

We discuss how such a system of labelled tableaux can be developed
from standard knowledge about tableaux (rules, branchs, closure)
\cite{Fitting1990,Smullyan1968} and from previous works on labelled tableaux for BI logic and some extensions \cite{GMP2005,GKP2020,Larchey-Wendling2014}. The main questions  concern the definition of labels, of labelled formulas, of sets of label constraints, in such a way one can define closure conditions for a tableaux branch and also what is a tableaux proof for the given logic. 
As the labels and constraints have to capture the semantics of the logic then the main challenge in our setting is to use the structure of the logic's models --- here our system models --- to determine a system of labels to be used to express the constraints one has to assert or to check when tableaux rules are applied. \\

The set-up of the logic presented in Section~\ref{sec:logic} is more
complicated than previously studied logical set-ups, and providing a
labelled tableaux system for our new logic presents more challenges.  
It is necessary and non-trivial to identify an appropriate algebra of
labels that captures the semantics specified by the satisfaction  
relation 
\[
    f , \mathcal{V} , x \models \phi 
\]
as defined in Section~\ref{sec:logic}, where $f$ is a
system, $\mathcal{V}$ is a valuation, and $x\in\Beh{f}$. Unpacking
this, we can see some of the challenges. Recall that, for every basic
component $c\in\C$, we assume a set $Var(c)$ of variables, and a
\emph{valuation} is a function $\val$ mapping each $p\in Var(c)$ to a
subset of $\Beh{c}$.  \\

The tableaux rules for the  $\wedge$ and $\vee$ connectives  can have the following form:  
\[
\begin{array}{c@{\qquad}c}
\dfrac{(\mathbb{T} (\phi \wedge \psi),l)}
    {\begin{array}{c}
    (\mathbb{T} (\phi),l) \\ 
    (\mathbb{T} (\psi),l)
    \end{array}}\; \mathbb{T}\wedge 
    & \dfrac{(\mathbb{F} (\phi \wedge \psi),l)}
        {(\mathbb{F}(\phi),l) \mid (\mathbb{F}(\psi),l)} \; \mathbb{F}\wedge  
\end{array}
\]

\[
\begin{array}{c@{\qquad}c}
\dfrac{(\mathbb{T} (\phi \vee \psi),l)}
        {(\mathbb{T}(\phi),l) \mid (\mathbb{T}(\psi),l)} \; \mathbb{T} \vee 
& 
\dfrac{(\mathbb{F} (\phi \vee \psi),l)}
    {\begin{array}{c}
    (\mathbb{F} (\phi),l) \\ 
    (\mathbb{F} (\psi),l)
    \end{array}}\; \mathbb{F} \vee
\end{array}
\]
Note that there is no change of label $l$ when rules are applied. If we look at the semantics for these connectives we can propose to have $l$ being of the form $(f,x)$. 

Structural connectives, such as our $\ast$ and $\stackrel{\rightarrow}{\ast}$, typically refer to labels that  correspond to components of decomposition of the systems that support the conjunctions $\phi \ast \psi$ and $\phi \stackrel{\rightarrow}{\ast} \psi$ in Section~\ref{subsec:struct-conn}.
Tableaux systems have been given for various versions of BI (e.g.,
\cite{GMP2005,Docherty2019,GKP2020}) dealing with the connectives $\ast$ and $\wand$ . In that setting, the multiplicative conjunction, $\ast$, can have the following rules:
\[
    \dfrac{\langle \mathbb{T}(\phi \ast \psi) , l \rangle}
        { \begin{array}{c} \langle \mathbb{T}(\phi), l_1 \rangle \\ 
             \langle \mathbb{T}(\phi), l_2 \rangle \end{array}
             }
           \; S(l,l_1,l_2) \; \mathbb{T}\ast
            \quad\mbox{and}\quad           
         \dfrac{\langle \mathbb{F}(\phi \ast \psi) , l \rangle}
       {\langle \mathbb{F}(\phi), l_1 \rangle \mid \langle
        \mathbb{F}(\psi) , l_2 \rangle} \; R(l,l_1,l_2) \;
      \mathbb{F}\ast
\]

Here the rule $\mathbb{T}\ast$ introduces new labels $l_1$ and $l_2$ and the label constraint (or assertion) $S(l,l_1,l_2)$ that is satisfied and the  rule $\mathbb{F}\ast$ can be applied only if there exist $l_1$ and $l_2$ among the existing labels that satisfy the label constraint (or requirement) $R(l,l_1,l_2)$ from a specific closure of the set of assertions. The labels, their composition (with and operator $\bullet$) and the label constraints allows us to capture the relational semantics of $\ast$ (see \cite{GMP2005,DP2018,Docherty2019}).  

In order to consider $\ast$ in our logic, we can propose that $l = (f,x)$, $l_1 = (f_1,x_1)$ and $l_2 = (f_2,x_2)$, and $S(l,l_1,l_2)$ and  $R(l,l_1,l_2)$ of the form $l \ = \ l_1 \bullet l_2$. The conditions on the $f_i$s in the definition of $\ast$ are essential local and then we can expect to find a way to express the constraints and to reason on their satisfaction and resolution. 

The interface-mediated $\stackrel{\rightarrow}{\ast}$ rule requires  that the labelling algebra maintain dependencies across different branches of the tableau as specified by the condition that `$\Int{f_1}{f_2}$ is an input to $f_2$'. This is a relatively complex and delicate development of the underlying approach to defining tableaux systems, which we defer to another occasion. 

Finally modalities, such as our $\Box$, typically require a relationship on labels that corresponds to an accessibility relation on worlds. In our case, accessibility is generated  by behaviours.

\[
    \frac{\langle \mathbb{T}(\Box\phi) , l \rangle}
        {\langle \mathbb{T}(\phi) , m \rangle} 
            \quad S(l,m) \quad \mathbb{T}\Box  
    \quad\mbox{and}\quad 
    \frac{\langle \mathbb{F}(\Box\phi) , l \rangle}
        {\langle \mathbb{F}(\phi) , m \rangle} 
            \quad R(l,m) \quad \mathbb{F}\Box  
\]

Here, if we look the semantics of our logic, we could consider $l \ = \ (f,x)$ and $m \ = \ (f,y)$  with, respectively, the assertion $S(l,m)$ and the requirement $R(l,m)$ expressing a relation between $x$ and $y$ that corresponds to the accessibility relation associated to our $\Box$. 

All these points illustrate that the definitions of labels and label constraints, and their treatment and resolution during a tableau construction, require a deeper analysis and a relatively complex and delicate development in order to propose a tableaux calculus that will be proved sound and complete. Such questions will be explored in future work.

\subsection{Other}

Other future work might include systematic approaches to extensions of the theory, while staying true to the minimalistic spirit in which the present work has been carried out. For example, two general notions, which have not been discussed, are \emph{abstraction} and \emph{refinement} of systems \cite{butler2013mastering} together with \emph{substitution} (see, for example, \cite{CIP2022}) of models. The theory needed to support these ideas will necessarily require an extension of our Hennessy-Milner theorem (Theorem~\ref{thm:HM}) to include the structural connectives. The development of such a theory is a substantial project for further work. 

Other theoretical work might include the incorporation of causality concepts 
(and associated notions) 
into the minimalistic framework. Again, this is further  work. 

More applied work might  consider the integration stochastic behaviours, execution and simulations, and model-checking, so providing the basis for the framework to account for --- that is, give semantics for --- a wide range of system modelling approaches.

\section*{Acknowledgements}

This work has been partially supported by the UK EPSRC research 
grants EP/S013008/1 and EP/R006865/1. The authors are grateful 
to Tristan Caulfield for helpful discussions of the systems background to the ideas presented.

\bibliographystyle{splncs04}
\bibliography{mybib}

\appendix

\section{Proof of recovering a standard form of the CAP theorem} \label{app:CAP}

We start with an informal presentation and proof of CAP. Consider the situation where there are two replicated copies of a database. User 1 and 2 interact with their respective copy of the database by overwriting the dataset, or by requesting to read the current value of the data. In this scenario, the CAP guarantees can be formulated as follows:
\begin{itemize}
    \item[-] \emph{consistency}: Any read returns the content of the latest write (by either user) at time of the submission of the read request, or any later write.
    \item[-] \emph{availability}: Any read and write requests are eventually performed.
    \item[-] \emph{partition-tolerance}: The system tolerates arbitrary message loss between the two databases.
\end{itemize}

A semi-formal proof that these guarantees cannot be satisfied simultaneously runs as follows: the current value of the database, and assume the connection between the two copies of the database breaks down. Now imagine user~2 writing the data $\chi'\neq \chi$ on the database, followed by user~1 requesting to read the data. By the availability and partition-tolerance requirements, the system must eventually return an answer $\chi''$ to user~1. Clearly, if $\chi''\neq\chi'$ then consistency is violated as the return value is outdated. So assume $\chi''=\chi'$. But now, as the first copy of the database has lost its connection to the second one, it cannot distinguish between the situation where the write of user~2 has occurred and a situation where no such write occurred. Thus, $\chi'$ will also be the return value for user~1 in case the write didn't occur, and in that situation $\chi'$ would be an outdated value thus violating consistency. 

We now formalize this scenario to show that the impossibility result is an instance of Theorem~\ref{thm:CAP}. In fact, by writing out the proof of Theorem~\ref{thm:CAP} with the formalization presented below one essentially obtains the proof in~\cite{GL2002}. For $i=1,2$ consider a language $\Sigma_i$ that contains for every timestamp $t\in\mathbb{Q}^+_0$ and string $s$ one of the atoms below left, with their intended interpretation given on the right:
\begin{center}
\begin{tabular}{l c l}
    $t:rd^i?$ & $\ldots$ & at time $t$, user $i$ requests a read \\
    $t:rd^i(s)$ & $\ldots$ & at time $t$, user $i$ reads value $s$\\
    $t:wr^i(s)$ & $\ldots$ & at time $t$, user $i$ writes value $s$
\end{tabular}
\end{center}

A sequence in $(\Sigma_1\cup\Sigma_2)^*$ is \emph{well-formed} if its timestamps are strictly ascending. 
Let $\mathcal{S}$ denote the set of well-formed sequences, and let $\mathcal{S}_i:=\mathcal{S}\cap\Sigma_i^*$.

We consider a composition $g_1\coimpl[\sigma_1]f\impl[\sigma_2]g_2$ where the behaviour sets of $g_1,g_2$ and $f$ are $\mathcal{S}_1,\mathcal{S}_2$ and $\mathcal{S}$, respectively. The implementation $\sigma_i$ is the obvious projection function that removes from a sequence in $\mathcal{S}$ all actions not initiated by the user $i$. The \emph{current value at $t$} of a sequence in $\mathcal{S}$ is the value of the latest write in it (by either user) with timestamp $\leq t$; if no such write exists, the current value equals some fixed string $s_0$ that we call the \emph{initial value}. A sequence in $\mathcal{S}$ is \emph{consistent} if whenever $t:rd^i?$ and $t':rd^i(s)$ are a pair of actions in the sequence with $t<t'$ and no read by $i$ in between, then $s$ coincides with the current value at some timestamp $\geq t$ and $\leq t'$. Let $C$ be the set of all consistent sequences.

Define two binary relations on $\mathcal{S}$. $R$ relates any sequence to the same sequence appended with $t:wr^2(s)$, for some string $s$ (and sufficiently large $t$). $S$ relates a sequence to the same sequence appended with $t:rd^1?$, a finite string of actions by user $2$, and a final $t':rd^1(s)$, for some string $s$. So, if an implementation of $f$ satisfies strong $R$-availability, it allows user $2$ to write anything at any time. If it satisfies weak $S$-satisfiability, it allows user $1$ to send a request anytime and wait until it is 
answered.

So, clearly $C$-consistency, strong $R$-, and weak $S$-availability are are required guarantees in the CAP scenario. We will now show that these are incompatible with $(\sigma_1,\sigma_2)$-partition-tolerance. By Theorem~\ref{thm:CAP}, it suffices to demonstrate the entanglement condition. Indeed, let $w\in\mathcal{L}$ and let $x\in R(w,\bullet)$ an extension of $w$ with $t:wr^2(s)$ where $s$ is neither $s_0$ nor the value of any write in $w$. Now consider any $v\in S(x,\bullet)$. Such $v$ is of the form $v=w\smallfrown\langle t:wr^2(s),t':rd^1?,\ldots,t'':rd^1(s')\rangle$, where the dotted part contains only actions by user $2$. Let $y,z\in\mathcal{L}$ satisfy $\sigma_1(y)=\sigma_1(z)=\sigma_1(v)$, $\sigma_2(y)=\sigma_2(w)$ and $\sigma_2(z)=\sigma_2(x)$. Therefore, $y=w\smallfrown\langle t':rd^1?,t'':rd^1(s')\rangle$ and
$z=w\smallfrown\langle t:wr^2(s),t':rd^1?,t'':rd^1(s')\rangle$.
We have to show $y\notin C$ or $z\notin C$. Let us assume $z\in C$. Then $s'=s$. Since, by construction, $s$ is neither $s_0$ nor the value of any write in $w$, it follows that $y\notin C$. This demonstrates the entanglement condition.

\newpage

\end{document}